\documentclass[a4paper,11pt]{amsart}

\usepackage{amssymb,amscd,amsmath}

\theoremstyle{plain}

\newtheorem{thm}{Theorem}[section]
\newtheorem{prop}[thm]{Proposition}
\newtheorem{lem}[thm]{Lemma}

\newtheorem{con}[thm]{Conjecture}

\theoremstyle{definition}

\theoremstyle{remark}

\newcommand{\forme}[1]{}

\title{The Italian domination numbers of some products of directed cycles}

\author[Kim]{Kijung Kim}
\address{Department of Mathematics, Pusan National University, Busan 46241, Republic of Korea}
\email{knukkj@pusan.ac.kr}

\date{\today}
\begin{document}

\begin{abstract}
An Italian dominating function on a digraph $D$ with vertex set $V(D)$
is defined as a function $f : V(D) \rightarrow \{0, 1, 2\}$  such that every vertex $v \in V(D)$ with
$f(v) = 0$ has at least two in-neighbors assigned $1$ under $f$ or one in-neighbor $w$ with $f(w) = 2$.
In this paper, we determine the exact values of the Italian domination numbers of some products of directed cycles.
\bigskip

\end{abstract}

\maketitle

\insert\footins{\footnotesize
This research was supported by Basic Science Research Program through the National Research Foundation of Korea funded by the Ministry of Education (2020R1I1A1A01055403).}

\section{Introduction and preliminaries}\label{sec:intro}

Let $D=(V,A)$ be a finite simple digraph with vertex set $V=V(D)$ and arc set $A=A(D)$.
An arc joining $v$ to $w$ is denoted by $v \rightarrow w$.
The \textit{maximum out-degree} and \textit{maximum in-degree} of $D$ are denoted by $\Delta^+(D)$ and $\Delta^-(D)$, respectively.

Let $D_1 = (V_1,A_1)$ and $D_2=(V_2, A_2)$ be two digraphs.
The \textit{cartesian product} of $D_1$ and $D_2$ is the digraph $D_1 \square D_2$ with vertex set $V_1 \times V_2$ and for two vertices $(x_1,x_2)$ and $(y_1,y_2)$,
\[(x_1,x_2) \rightarrow (y_1,y_2)\]
if one of the following holds:
(i) $x_1 = y_1$ and $x_2 \rightarrow y_2$;
(ii) $x_1 \rightarrow y_1$ and $x_2=y_2$.

The \textit{strong product} of $D_1$ and $D_2$ is the digraph $D_1 \otimes D_2$ with vertex set $V_1 \times V_2$ and for two vertices $(x_1,x_2)$ and $(y_1,y_2)$,
\[(x_1,x_2) \rightarrow (y_1,y_2)\]
if one of the following holds:
(i) $x_1 \rightarrow y_1$ and $x_2 \rightarrow y_2$;
(ii) $x_1 = y_1$ and $x_2 \rightarrow y_2$;
(iii) $x_1 \rightarrow y_1$ and $x_2=y_2$.

An \textit{Italian dominating function} (IDF) on a digraph $D$
is defined as a function $f : V(D) \rightarrow \{0, 1, 2\}$  such that every vertex $v \in V(D)$ with
$f(v) = 0$ has at least two in-neighbors assigned $1$ under $f$ or one in-neighbor $w$ with $f(w) = 2$.
In other words, we say that a vertex $v$ for which $f(v) \in \{1, 2\}$ dominates itself, while a vertex $v$ with $f(v)=0$
is dominated by $f$ if it has at least two in-neighbors assigned $1$ under $f$ or one in-neighbor $w$ with $f(w) = 2$.
An Italian dominating function $f : V(D) \rightarrow \{0, 1, 2\}$ gives a partition $\{V_0, V_1, V_2\}$ of $V(D)$, where
$V_i:= \{ x \in V(D) \mid f(x)=i \}$.
The \textit{weight} of an Italian dominating function $f$ is the value $\omega(f) = f(V(D)) = \sum_{u \in V(D)} f(u)$.
The \textit{Italian domination number} of a digraph $D$, denoted by $\gamma_I(D)$, is the minimum taken over the weights of all Italian dominating functions on $D$.
A \textit{$\gamma_I(D)$-function} is an Italian dominating function on $D$ with weight $\gamma_I(D)$.

The Italian dominating functions in graphs and digraphs have studied in \cite{ CHHM, HK, GWLY, GXLZY, Vol, Vol-2}.
The authors of \cite{CHHM} introduce the concept of Italian domination and present bounds relating the Italian domination number to some other domination parameters.
The authors of \cite{HK} characterize the trees $T$ for which $\gamma(T) +1 =\gamma_I(T)$ and also
characterize the trees $T$ for which $\gamma_I(T) =2\gamma(T)$.
After that, there are some studies on the cartesian products of undirected cycles or undirected paths in \cite{GWLY, GXY, LSX, SSSM}.
Recently, the author of \cite{Vol} initiated the study of the Italian domination number in digraphs.
In this paper, we investigate the Italian domination numbers of cartesian products and strong products of directed cycles.

The following results are useful to our study.

\begin{prop}[\cite{Vol}]\label{mainthm2}
Let $D$ be a digraph of order $n$. Then $\gamma_I(D) \geq \lceil \frac{2n}{2 + \Delta^+(D)} \rceil$.
\end{prop}

\begin{prop}[\cite{Vol}]\label{mainthm3}
Let $D$ be a digraph of order $n$. Then $\gamma_I(D) \leq n$  and  $\gamma_I(D) = n$ if and only if $\Delta^+(D), \Delta^-(D) \leq 1$.
\end{prop}

\begin{prop}[\cite{Vol}]\label{mainthm3-cor}
If $D$ is a directed path or a directed cycle of order $n$, then $\gamma_I(D) = n$.
\end{prop}

\section{The Italian domination numbers of some products of directed cycles}\label{sec:main1-2}

In this section, we determine the exact values of the Italian domination numbers of some products of directed cycles.

First, we consider the cartesian product of directed cycles.
We denote the vertex set of a directed cycle $C_m$ by $\{1,2,\dotsc ,m\}$, and
assume that $i \rightarrow i+1$ is an arc of $C_m$.
For every vertex $(i,j) \in V(C_m \square C_n)$, the first and second components are considered modulo $m$ and $n$, respectively.
For each $1 \leq k \leq n$, we denote by $C_m^k$ the subdigraph of $C_m \square C_n$ induced by the set $\{ (j,k) \mid 1 \leq j \leq m\}$.
Note that $C_m^k$ is isomorphic to $C_m$.
Let $f$ be a $\gamma_I(C_m \square C_n)$-function and set $a_k = \sum_{x \in V(C_m^k)} f(x)$.
Then $\gamma_I(C_m \square C_n) = \sum_{k=1}^n a_k$.
It is easy to see that $C_m \square C_n$ is isomorphic to $C_n \square C_m$. So, $\gamma_I(C_m \square C_n) = \gamma_I(C_n \square C_m)$.

\begin{thm}\label{mainthm3}
If $m=2r$ and $n=2s$ for some positive integers $r, s$, then $\gamma_I(C_m \square C_n) = \frac{mn}{2}$.
\end{thm}
\begin{proof}
Define $f : V(C_m \square C_n) \rightarrow \{0,1,2\}$ by
\[f((2i-1,2j-1)) = f((2i,2j)) = 1\]
for each $1 \leq i \leq r$ and $1 \leq j \leq s$,
and \[f((x, y))=0\] otherwise.
It is easy to see that $f$ is an IDF of $C_m \square C_n$ with weight $\frac{mn}{2}$ and
so $\gamma_I(C_m \square C_n) \leq  \frac{mn}{2}$.
Since $\Delta^+(D)=2$, it follows from Proposition \ref{mainthm2} that $\gamma_I(C_m \square C_n) \geq  \frac{mn}{2}$.
Thus, we have $\gamma_I(C_m \square C_n) = \frac{mn}{2}$.
\end{proof}

\begin{thm}\label{mainthm5}
For an odd integer $n \geq 3$,
$\gamma_I(C_2 \square C_n) = n+1$.
\end{thm}
\begin{proof}
Define $f : V(C_2 \square C_n) \rightarrow \{0,1,2\}$ by
\[f((1,2j-1)) =1\] for each $1 \leq j \leq \frac{n+1}{2}$,
\[f((2,2j)) =1\] for each $1 \leq j \leq \frac{n-1}{2}$,
\[f((2,n)) =1\] and \[f((x, y))=0\] otherwise.
It is easy to see that $f$ is an IDF of $C_2 \square C_n$ with weight $n+1$ and
so $\gamma_I(C_2 \square C_n) \leq  n+1$.

Now we claim that $\gamma_I(C_2 \square C_n) \geq n+1$.
Suppose to the contrary that $\gamma_I(C_2 \square C_n) \leq  n$.
Let $f$ be a $\gamma_I(C_2 \square C_n)$-function.
If $a_k=0$ for some $k$, assume without loss of generality $k=3$, then $f((1,3))=f((2,3))=0$.
To dominate the vertices $(1,3)$ and $(2,3)$, we must have $f((1,2)) = f((2,2)) = 2$.
Define $g : V(C_2 \square C_n) \rightarrow \{0,1,2\}$ by \[g((1,2))= g((2,1))= g((2,3))=1, g((2,2))=0\]
and \[g((x,y))= f((x,y))\] otherwise.
Then $g$ is an IDF of $C_2 \square C_n$ with weight less than $\omega(f)$, which is a contradiction.
Thus, $a_k \geq 1$ for each $k$.
By assumption, $a_k =1$ for each $k$.
Without loss of generality, we assume that $f((1,2))=1$.
To dominate $(2,2)$, we must have $f((2,1))=1$. Since $a_3=1$ and $f((2,2))=0$, we have $f((2,3))=1$.
By repeating this process, we obtain $f((1,2i))=1$ for each $1 \leq i \leq \frac{n-1}{2}$, $f((2,2i-1))=1$ for $1 \leq i \leq\frac{n+1}{2}$
and $f((x,y))=0$ otherwise. But, the vertex $(1,1)$ is not dominated, a contradiction.
Thus we have $\gamma_I(C_2 \square C_n) \geq n+1$.
This completes the proof.
\end{proof}

\begin{thm}\label{mainthm7}
For an integer $n \geq 3$,
$\gamma_I(C_3 \square C_n) = 2n$.
\end{thm}
\begin{proof}
When $n=3r$ for some positive integer $r$, define $f_0 : V(C_3 \square C_n) \rightarrow \{0,1,2\}$ by
\[f_0((1,3j+1)) = f_0((2,3j+1)) = 1\]
for each  $0 \leq j \leq n-1$,
\[f_0((2,3j+2)) = f_0((3,3j+2)) = 1\]
for each  $0 \leq j \leq n-1$,
\[f_0((1,3j+3)) = f_0((3,3j+3)) = 1\]
for each  $0 \leq j \leq n-1$ and
\[f_0((x,y))=0\]
otherwise.

When $n=3r+1$ for some positive integer $r$, define $f_1 : V(C_3 \square C_n) \rightarrow \{0,1,2\}$ by
\[f_1((2,n)) = f_1((3,n)) = 1\]
and
\[f_1((x,y)) = f_0((x,y))\]
otherwise.

When $n=3r+2$ for some positive integer $r$, define $f_2 : V(C_3 \square C_n) \rightarrow \{0,1,2\}$ by
\[f_2((1,n-1)) = f_2((2,n-1)) = f_2((1,n)) = f_2((3,n)) = 1\]
and
\[f_2((x,y)) = f_0((x,y))\]
otherwise.
It is easy to see that $f_i$ $(i=0,1,2)$ is an IDF of $C_3 \square C_n$ with weight $2n$
and so $\gamma_I(C_3 \square C_n) \leq 2n$.

\vskip5pt
Now we prove that $\gamma_I(C_3 \square C_n) \geq 2n$.
Let $f$ be a $\gamma_I(C_3 \square C_n)$-function.

\vskip5pt
\textbf{Claim 1.} $a_k \geq 1$ for each $1 \leq k \leq n$.

\text{Proof.}
Suppose to the contrary that $a_k =0$ for some $k$, say $k=n$.
To dominate $(1,n)$, $(2,n)$ and $(3,n)$, we must have $f((1,n-1))=f((2,n-1))=f((3,n-1))=2$.
But, the function $g$ defined by
\[g((1,n-1))=g((2,n-1))=g((3,n-1))=1,\]
\[g((1,n))=g((2,n))=1\]
and
\[g((x,y)) = f((x,y))\]
otherwise, is an IDF of $C_3 \square C_n$ with weight less than $\omega(f)$.
This is an contradiction.
\qed

\vskip5pt
We choose a $\gamma_I(C_3 \square C_n)$-function $h$ so that the size of $M_h:=\{ k \mid a_k=1\}$ is as small as possible.
\vskip5pt
\textbf{Claim 2.} $|M_h|=0$.

\text{Proof.}
Suppose to the contrary that $|M_h| \geq 1$.
Without loss of generality, assume that $a_n =1$ and $h((1,n))=1$.
To dominate $(2,n)$ and $(3,n)$, we must have $h((2,n-1))=1$ and $h((3,n-1))=2$.
If $n=3$, then clearly $a_1 \geq 2$ and so $\gamma_I(C_3 \square C_3) \geq 6$.
However, when $n=3$, the previously defined function $f_0$ is a $\gamma_I(C_3 \square C_3)$-function such that $\omega(f_0)=6$ and $|M_{f_0}|=0$.
This contradicts the choice of $h$.
From now on, assume $n \geq 4$.
We divide our consideration into the following two cases.

\vskip5pt
\text{Case 1.} $a_{n-2} =1$.

By the same argument as above, we have $a_{n-3} \geq 3$. So $a_{n-3} + a_{n-2} + a_{n-1} + a_n \geq 8$.
If $n=4$, then the previously defined function $f_1$ induces a contradiction. Suppose $n \geq 5$.
Since $a_{n-4} \geq 1$ by Claim 1, $h((i,n-4))= 1$ or $2$ for some $i \in \{1,2,3\}$.
Without loss of generality, we may assume $h((1,n-4))= 1$ or $2$.
Define $t : V(C_3 \square C_n) \rightarrow \{0,1,2\}$ by
\[t((1,n-3)) = t((2,n-2)) = t((1,n-1)) = t((2,n)) =0,\]
\[t((2,n-3)) = t((1,n-2)) = t((2,n-1)) = t((1,n)) =1,\]
\[t((3,n-3)) = t((3,n-2)) = t((3,n-1)) = t((3,n)) =1\]
and
\[t((x,y))= h((x,y))\]
otherwise.
Then it is easy to see that $t$ is an IDF of $C_3 \square C_n$ such that $|M_t| < |M_h|$.
This contradicts the choice of $h$.

\vskip5pt
\text{Case 2.} $a_{n-2} \geq 2$.

Now $a_{n-2} + a_{n-1} + a_n \geq 6$.
Since $a_{n-3} \geq 1$ by Claim 1, $h((i,n-3))= 1$ or $2$ for some $i \in \{1,2,3\}$.
Without loss of generality, we may assume $h((1,n-3))= 1$ or $2$.
Define $t : V(C_3 \square C_n) \rightarrow \{0,1,2\}$ by
\[t((1,n-2)) = t((2,n-1)) = t((3,n)) = 0,\]
\[t((2,n-2)) = t((1,n-1)) = t((1,n)) = 1,\]
\[t((3,n-2)) = t((3,n-1)) = t((2,n)) = 1\]
and
\[t((x,y))= h((x,y))\]
otherwise.
Then it is easy to see that $t$ is an IDF of $C_3 \square C_n$ such that $|M_t| < |M_h|$.
This contradicts the choice of $h$.
\qed

\vskip5pt
By Claims 1 and 2, we have $\gamma_I(C_3 \square C_n) \geq 2n$.
This completes the proof.
\end{proof}

\vskip15pt

Next, we consider the strong product of directed cycles.
We denote the vertex set of a directed cycle $C_m$ by $\{1,2,\dotsc ,m\}$, and
assume that $i \rightarrow i+1$ is an arc of $C_m$.
For every vertex $(i,j) \in V(C_m \otimes C_n)$, the first and second components are considered modulo $m$ and $n$, respectively.
For each $1 \leq k \leq n$, we denote by $C_m^k$ the subdigraph of $C_m \otimes C_n$ induced by the set $\{ (j,k) \mid 1 \leq j \leq m\}$.
Note that $C_m^k$ is isomorphic to $C_m$.
Let $f$ be a $\gamma_I(C_m \otimes C_n)$-function and set $a_k = \sum_{x \in V(C_m^k)} f(x)$.
Then $\gamma_I(C_m \otimes C_n) = \sum_{k=1}^n a_k$.

\begin{lem}\label{bound-strong}
For positive integers $m, n \geq 2$, $\gamma_I(C_m \otimes C_n) \geq \lceil \frac{mn}{2} \rceil$.
\end{lem}
\begin{proof}
Note that the vertices of $C_m^k$ are dominated by vertices of $C_m^{k-1}$ or $C_m^k$.
It suffices to verify that $\sum_{k=1} ^n  a_k \geq \lceil \frac{mn}{2} \rceil$.
In order to do, we claim $a_k + a_{k+1} \geq m$ for each $k$.
First of all, we assume that $a_{k+1}=0$.
Then to dominate $(i, k+1)$ for each $1 \leq i \leq m$, we must have
\[f((i-1,k)) + f((i,k)) \geq 2.\]
Then $2a_k = \sum_{i=1} ^m (f((i-1,k)) + f((i,k))) \geq 2m$ and hence $a_k + a_{k+1}  \geq m$.
If $a_{k+1} = t >0$, then there will be at least $m-t$ vertices in $V_0$ that will be dominated only by vertices of $C_m^k$.
This fact induces $a_k \geq m-t$ and so $a_k +a_{k+1} \geq m$.
Therefore, we have
\[2\gamma_I(C_m \otimes C_n) = 2\sum_{k=1} ^n  a_k = \sum_{k=1} ^n  (a_k + a_{k+1}) \geq nm.\]
This completes the proof.
\end{proof}

\begin{thm}\label{mainthm6}
For positive integers $m, n \geq 2$, $\gamma_I(C_m \otimes C_n)= \lceil \frac{mn}{2} \rceil$.
\end{thm}
\begin{proof}
We divide our consideration into the following two cases.

\vskip5pt
\text{Case 1.} $m$ or $n$ is even.

Since $C_m \otimes C_n$ is isomorphic to $C_n \otimes C_m$, we may assume that $n=2s$ for some positive integers $s$.

Define $f : V(C_m \otimes C_n) \rightarrow \{0,1,2\}$ by
\[f((i,2j-1))=1\]
for each $1 \leq i \leq m$ and $1 \leq j \leq s$, and
\[f((x,y))=0\]
otherwise.
It is easy to see that $f$ is an IDF of $C_m \otimes C_n$ with weight $\frac{mn}{2}$ and
so $\gamma_I(C_m \otimes C_n) \leq \lceil \frac{mn}{2} \rceil$.
Thus, it follows from Lemma \ref{bound-strong} that $\gamma_I(C_m \otimes C_n)= \lceil \frac{mn}{2} \rceil$.

\vskip5pt
\text{Case 2.}  $m=2r+1, n=2s+1$ for some positive integers $r, s$.

Define $f : V(C_m \otimes C_n) \rightarrow \{0,1,2\}$ by
\[f((2i+1,2j+1))=1\]
for each $0 \leq i \leq r$ and $0 \leq j \leq s$,
\[f((2i,2j))=1\] for each $1 \leq i \leq r$ and $1 \leq j \leq s$ and
\[f((x,y))=0\] otherwise.
It is easy to see that $f$ is an IDF of $C_m \otimes C_n$ with weight $(r+1)(s+1)+rs$ and
so $\gamma_I(C_m \otimes C_n) \leq \lceil \frac{mn}{2} \rceil$.
Thus, it follows from Lemma \ref{bound-strong} that $\gamma_I(C_m \otimes C_n)= \lceil \frac{mn}{2} \rceil$.
\end{proof}

\section{Conclusions}\label{sec:con}
In this paper, we determined the exact values of $\gamma_I(C_2 \square C_l)$, $\gamma_I(C_3 \square C_l)$ and $\gamma_I(C_m \square C_n)$
for an integer $l$ and even integers $m, n$.
The other cases are still open. We conclude by giving a conjecture.

\begin{con}
For an odd integer $n$, $\gamma_I(C_4 \square C_n) = 2n+2$.
\end{con}

\bibstyle{plain}

\end{document}